\documentclass{cc}

\newcommand{\NN}{\mathbb N}
\newcommand{\ZZ}{\mathbb Z}
\newcommand{\BB}{\mathbb B}
\newcommand{\frS}{\mathfrak S}
\newcommand{\vv}{\mathbf}
\newcommand{\LC}{L^\text{(C)}}
\newcommand{\LT}{L^\text{(T)}}
\newcommand{\WC}{W^\text{(C)}}
\newcommand{\WT}{W^\text{(T)}}

\contact{dmitriy.zakablukov@gmail.com}
\submitted{19 March 2016}
\title{On Asymptotic Gate Complexity and Depth of Reversible Circuits With Additional Memory}
\titlehead{Gate Complexity and Depth of Reversible Logic}
\author{
    Dmitry V. Zakablukov\\
    Department of Information Security\\
    Bauman Moscow State Technical University\\
    Moscow, Russian Federation 105005\\
    \email{dmitriy.zakablukov@gmail.com}
}

\begin{abstract}
    The reversible logic can be used in various research areas, e.\,g. quantum computation, cryptography and signal processing.
    In the paper we study reversible logic circuits with additional inputs,
    which consist of NOT, CNOT and C\textsuperscript{2}NOT gates.
    We consider a set $F(n,q)$ of all transformations $\BB^n \to \BB^n$ that can be realized by reversible circuits
    with $(n+q)$ inputs.
    An analogue of Lupanov's method for the synthesis of reversible logic circuits with additional inputs is described.
    We prove upper asymptotic bounds for the Shannon gate complexity function $L(n,q)$ and the depth function $D(n,q)$
    in case of $q > 0$: $L(n,q_0) \lesssim 2^n$ if $q_0 \sim n 2^{n-o(n)}$ and $D(n,q_1) \lesssim 3n$ if $q_1 \sim 2^n$.
\end{abstract}


\begin{subject}
  03D15 Complexity of computation.
\end{subject}

\begin{document}

\section{Introduction}

The reversible logic is essential in the quantum computing. It also has a great potential in designing
various computing devices with low power consumption. \cite{landauer} proved that the irreversibility of
computations leads to the energy dissipation regardless of the underlying technology.
\cite{bennett} showed that the absence of heat generation can be achieved only when a circuit
is completely built from reversible gates.
The main problem is that we should find a compromise between the gate complexity,
the depth (working time) of a reversible circuit and the amount of used memory (additional inputs)
when solving the problem of reversible logic synthesis.
Unfortunately, strict asymptotic bounds for the gate complexity and the depth of reversible circuits
haven't been found so far, especially in the case of using additional inputs.

The circuit complexity theory goes back to the work of~\cite{shannon}. He was the first who suggested to consider
the complexity of the minimal switching circuit, which realizes a Boolean function, as a complexity measure of this function.
For today, the asymptotic gate complexity $L(n) \sim 2^n \mathop / n$ of a Boolean function of $n$ variables in the basis
of classical gates ``NOT, OR, AND'' is well-known.

The problem of computations with the limited memory was considered by~\cite{karpova}.
She proved that the asymptotic gate complexity of a circuit, which consists of the gates corresponding to all
Boolean functions of $p$ variables and which uses at least three memory registers,
depends on the value of $p$, but doesn't depend on the number of used memory registers.
Also she proved that any Boolean function can be realized in such a circuit using only two memory registers.

\cite{lupanov} considered circuits of functional elements with delays.
He proved that in a regular basis of functional elements any Boolean function can be realized in a circuit
with asymptotically the best gate complexity and with the delay $T(n) \sim \tau n$,
where $\tau$ is the constant depending on the basis.
Though the depth and the delay of a circuit can be defined differently (see~\citealt{khrapchenko}),
in the model of reversible circuit described below we can consider the value of $T(n)$
as the circuit depth. However a dependency of $T(n)$ on the number of used memory registers was not considered
for the ``classical'' circuits.

A gate is called reversible if it implements a bijective transformation.
There are several known reversible gates for today. Among them are the NOT gate; the controlled NOT (CNOT) gate,
introduced by~\cite{feynman}; the Toffoli gate (C\textsuperscript{2}NOT) introduced by~\cite{toffoli}; the Fredkin gate, etc.

A set $F(n,q)$ of all transformations $\BB^n \to \BB^n$
that can be implemented by reversible circuits with $(n+q)$ inputs was considered in~\cite{my_complexity_and_depth_no_memory}.
Also the Shannon gate complexity function $L(n,q)$ and the depth function $D(n,q)$ as functions of $n$
and the number of additional inputs $q$ (additional memory) were defined and upper bounded
in the case, when additional inputs are not allowed in a reversible circuit.

The subject of this paper is reversible logic circuits,
which consist of NOT, CNOT and C\textsuperscript{2}NOT gates and which can use an unlimited amount of additional inputs
(unlike the reversible circuits we have studied earlier, see~\citealt{my_complexity_and_depth_no_memory}).

We will describe an analogue of Lupanov's method for synthesizing a reversible circuit with additional inputs,
which has the minimal gate complexity or the minimal depth. Using this synthesis approach,
we will prove the following upper asymptotic bounds for the functions $L(n,q)$ and $D(n,q)$:
\begin{align*}
    L(n,q_0) &\lesssim 2^n \text{ , if \,} q_0 \sim n 2^{n-o(n)} \;  , \\
    D(n,q_1) &\lesssim 3n  \text{ , if \,} q_1 \sim 2^n \;  .
\end{align*}
Also, some upper bounds for the quantum weight function will be proved.

Using the lower and upper bounds for the functions $L(n,q)$ and $D(n,q)$,
we state that the usage of additional memory in a reversible circuit,
which consists of NOT, CNOT and C\textsuperscript{2}NOT gates,
almost always allows to reduce its gate complexity and the depth.

\section{Background}

The controlled NOT gate (CNOT) was introduced by~\cite{feynman}.
The Toffoli gate was introduced by~\cite{toffoli}.
The generalized Toffoli gate with multiple control inputs is usually denoted as C\textsuperscript{k}NOT
or TOF$_{k+1}$, where $k$ stands for the number of control inputs.
The synthesis of reversible circuits consisting of these gates
was discussed in several works, see~\cite{my_fast_synthesis_algorithm, iterative_compositions,
miller_spectral, miller_transform_based, saeedi_novel, maslov_rm_synthesis, saeedi_cycle_based}.

We use the following notation for a generalized Toffoli gate.
\begin{definition}
    A generalized Toffoli gate with $k$ control inputs
    $TOF^n_{k+1}(I;t) = TOF^n_{k+1}(i_1, \cdots, i_k, t)$ is a reversible gate with $n$ inputs,
    which defines a transformation $\BB^n \to \BB^n$ as follows:
    $$
        f_{TOF^n_{k+1}(I;t)}(\langle x_1, \cdots, x_n \rangle) = \langle x_1, \cdots,
            x_t \oplus x_{i_1}\wedge \cdots \wedge x_{i_k}, \cdots, x_n \rangle \;  ,
    $$
    where $I = \{\,i_1, \cdots, i_k \,\}$ is a set of indices of control input lines
    and $t$ is an index of a controlled output line, $t \notin I$.
\end{definition}

From the definition one can note that a gate TOF(a) is a NOT gate, TOF(a,b) is a CNOT gate
and TOF(a,b,c) is a C\textsuperscript{2}NOT gate.

We denote a set of all TOF$^n_{k+1}$ gates, $k < 3$, as $\Omega_n^2$ (i.\,e. all NOT, CNOT and C\textsuperscript{2}NOT gates).
An upper and/or a lower indices in TOF$^n_{k+1}$ will be omitted, if their value is clear from the context.

Fan-in, fan-out and a random connection of inputs and outputs of gates in a reversible circuit are forbidden.
We assume that all gates in a reversible circuit have exactly $n$ numbered inputs and outputs
and that the $i$-th output of a gate is connected only to the $i$-th input of the following gate.
Thus in our model of a reversible circuit a graph associated with a circuit presents itself a single chain.
We will refer to such a connection of reversible gates as \textit{composition}.

A symbol $r_i$ from a set $R = \{\,r_1, \cdots, r_n\,\}$ can be assigned to the $i$-th input and output of a gate.
All these symbols can be treated as names of memory registers (indices of memory cells),
which store the current computation result of a circuit.

If we consider all the gates from $\Omega_n^2$ regardless of an underlying technology, we can assume that they all have the same
technological cost. However, in a quantum technology, for example, a technological cost of NOT and CNOT gates is much less than
a technological cost of a Toffoli gate (see~\citealt{barenco}).
Hence, we will assume that a gate $e$ from $\Omega_n^2$ has the weight $W(e)$ depending on the underlying technology.
More precisely, we will asume that all NOT and CNOT gates from $\Omega_n^2$ have the same weight $\WC$ and
all C\textsuperscript{2}NOT gates from $\Omega_n^2$ have the weight $\WT$.

Let a reversible circuit $\frS$ with $n$ inputs be a composition of $l$ gates from $\Omega_n^2$:
$\frS = \mathop{*}_{j=1}^l {TOF(I_j; t_j)}$.
In the paper we study the following circuit's properties:
\begin{enumerate}
    \item
        The gate complexity $L(\frS)$, equal to the number of gates $l$.
        
    \item
        The quantum weight $W(\frS)$, equal to the sum of weights of all its gates.
        
    \item
        The depth $D(\frS)$, equal to the number of gates in the path from inputs to outputs that cannot be executed
        simultaneously.
\end{enumerate}

Note that the quantum weight $W(\frS)$ of a reversible circuit $\frS$ is not equal to its technological cost,
because they may significantly differ.
But we can state that in most cases a greater value of the function $W(\frS)$ means a greater
technological cost of a reversible circuit $\frS$.

A formal definition of the reversible circuit depth for our circuit model can be found
in~\cite{my_complexity_and_depth_no_memory}.
Here we just want to remind that a reversible circuit $\frS$ has the depth $D(\frS) = 1$, if for every two of its gates
$TOF(I_1; j_1)$ and $TOF(I_2; j_2)$ the following equation holds:
$$
    \left( \{\,t_1\,\} \cup I_1 \right) \cap \left( \{\,t_2\,\} \cup I_2 \right) = \emptyset \;  .
$$
Also, the depth $D(\frS)$ of a reversible circuit $\frS$ equals to the minimal number $d$ of disjoint sub-circuits
with the depth of each equal to one in the following equation:
$$
    \frS = \bigsqcup_{i=1}^d{{\frS'}_i}, \text{ } {\frS'}_i \subseteq \frS, \text{ }D({\frS'}_i) = 1 \;  .
$$

For example, a reversible circuit $\frS = TOF(1;2) * TOF(3,1) * TOF(2) * TOF(4) * TOF(1,4,2) * TOF(3)$
(see~\ref{pic_scheme_example}) has the gate complexity $L(\frS) = 6$ and the depth $D(\frS) = 3$,
because we can divide the circuit into three disjoint sub-circuits with the depth of each equal to one in the following manner:
$\frS = (TOF(1;2)) * (TOF(3,1) * TOF(2) * TOF(4)) * (TOF(1,4,2) * TOF(3))$.

\begin{figure}[ht]
    \centering
    \includegraphics{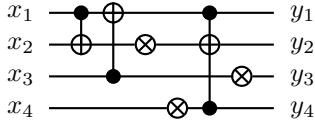}
    \caption
    {
        A reversible circuit $\frS = TOF(1;2) * TOF(3,1) * TOF(2) * TOF(4) * TOF(1,4,2) * TOF(3)$
        with the gate complexity $L(\frS) = 6$ and the depth $D(\frS) = 3$.
    }\label{pic_scheme_example}
\end{figure}

Note that the reversible circuit is equivalent to another one with the depth equal to three:
$\frS_1 = (TOF(1;2) * TOF(4)) * (TOF(3,1) * TOF(2)) * (TOF(1,4,2) * TOF(3))$. Therefore from here on we will consider
that such circuits are different in terms of our reversible circuit's model,
but equivalent in terms of the equality of Boolean transformations defined by them.

\section{Asymptotic bounds for reversible circuits without additional inputs}

It is well-known that a reversible circuit $\frS$ with $n \geq 4$ inputs defines an even permutation
on the set $\BB^n$, see~\cite{shende}.
But it can also implement a transformation $\BB^m \to \BB^k$, where $m, k \leq n$,
with or without an additional memory.
A circuit with $(n+q)$ inputs implements a transformation $f\colon \BB^n \to \BB^m$,
if there is such a permutation $\pi \in S(\ZZ_{n+q})$ for circuit outputs that every input of the form
$\langle x_1, \cdots, x_n, 0, \cdots, 0\rangle$ is transformed by the circuit into an output
$\langle y_1, \cdots, y_m, *, \cdots, *\rangle$ after applying the permutation $\pi$,
where $f(\langle x_1, \cdots, x_n\rangle) = \langle y_1, \cdots, y_m\rangle$
(see~\ref{pic_function_realization_with_memory}).

\begin{figure}[ht]
    \centering
    \includegraphics{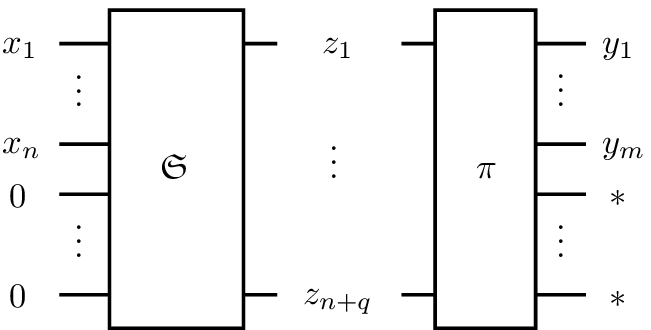}
    \caption
    {
        A reversible circuit $\frS$ implementing a transformation $f\colon \BB^n \to \BB^m$
        with $q$ additional inputs. For every $\vv x \in \BB^n$ the equation
        $f(\langle x_1, \cdots, x_n\rangle) = \langle y_1, \cdots, y_m\rangle$ holds.
    }\label{pic_function_realization_with_memory}
\end{figure}

We remind that in our terminology expressions ``implements a transformation'' and ``defines a transformation'' have
different meanings. If a circuit $\frS$ implements a transformation $f\colon \BB^n \to \BB^n$ and has exactly $n$ inputs,
we will say that this circuit implements $f$ \textit{without additional inputs}.

We marked all ``don't care'' outputs of a reversible circuit by the symbol~* on~\ref{pic_function_realization_with_memory}.
In most cases these outputs will not be cleared out in the end, i.\,e. they will contain a \textit{computational garbage}.
Unfortunately, this garbage can be removed only if a transformation $f$ implemented by a reversible circuit $\frS$
is bijective.
In this case we can clear out all garbage outputs, except ones corresponding to the inputs of $f$, with the help of
a part of the existing circuit (let's denote it as $\frS_*$). Then we can append a reversible circuit $\frS^{-1}$
implementing the transformation $f^{-1}$ with generating of computational garbage and with clearing out
the outputs corresponding to the inputs of $f$. And finally, we can remove
this generated garbage with the help of a part of the circuit $\frS^{-1}$ (let's denote it as $\frS^{-1}_*$).
Thus a resulting circuit $\frS_{\text{res}}$ will have the gate complexity
$L(\frS_{\text{res}}) \leq 4 \cdot \max(L(\frS), L(\frS^{-1}))$ and the depth
$D(\frS_{\text{res}}) \leq 4 \cdot \max(D(\frS), D(\frS^{-1}))$ (see~\ref{pic_clearing_out_garbage}).

\begin{figure}[ht]
    \centering
    \includegraphics{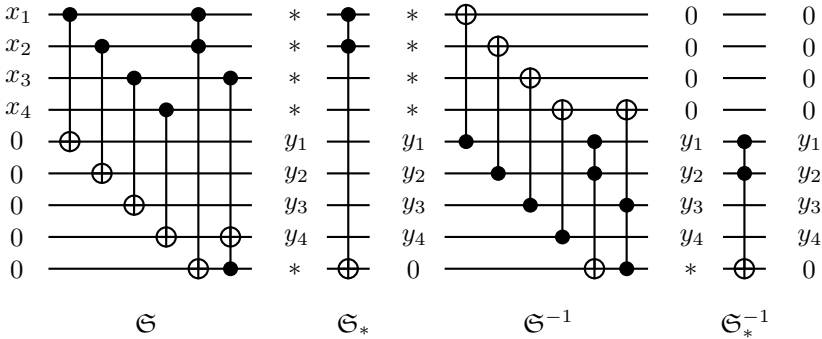}
    \caption
    {
        An example of a reversible circuit $\frS_{\text{res}} = \frS * (\frS_*) * \frS^{-1} * (\frS^{-1}_*)$
        with garbage removing.
    }\label{pic_clearing_out_garbage}
\end{figure}

Therefore, all asymptotic bounds for the gate complexity and the depth will be given later for a reversible circuit with
a computational garbage on the outputs. To obtain similar bounds for a reversible circuit without
a computational garbage on the outputs, one should multiply them by four.

Let $P_2(n,n)$ be the set of all transformations $\BB^n \to \BB^n$
and $F(n,q) \subseteq P_2(n,n)$ be the set of all transformations, which can be implemented by reversible circuits with
$(n+q)$ inputs.
It is not difficult to show that $F(n,0)$, $n > 3$, is equal to the set of transformations that are defined
by all the permutations from the alternating group $A(\BB^n)$ and $F(n,q) = P_2(n,n)$ if $q \geq n$.

We denote the minimum gate complexity, the minimum depth and the minimum quantum weight of a reversible circuit
among all reversible circuits implementing a transformation $f \in F(n,q)$ with $q$ additional inputs
as $L(f,q)$, $D(f,q)$ and $W(f,q)$ respectively. The Shannon gate complexity function $L(n,q)$,
the depth function $D(n,q)$ and the quantum weight function $W(n,q)$ are defined as follows:
\begin{align*}
    L(n,q) &= \max_{f \in F(n,q)} {L(f,q)} \;  , \\
    D(n,q) &= \max_{f \in F(n,q)} {D(f,q)} \;  , \\
    W(n,q) &= \max_{f \in F(n,q)} {W(f,q)} \;  .
\end{align*}

For the purpose of estimating the function $W(n,q)$, we will count the number of NOT/CNOT and C\textsuperscript{2}NOT gates in
a reversible circuit separately.
If we denote the number of NOT and CNOT gates in a reversible circuit $\frS$ as $\LC(\frS)$
and the number of C\textsuperscript{2}NOT gates as $\LT(\frS)$, then the following equation holds:
\begin{equation}\label{formula_quantum_weight_and_gate_complexity_dependency}
    W(\frS) = \WC \cdot \LC(\frS) + \WT \cdot \LT(\frS) \;  .
\end{equation}

We proved (see~\citealt{my_complexity_and_depth_no_memory}) that there is such $n_0 \in \NN$
that for $n > n_0$ the following equations hold:
\begin{align}
    L(n,q) &\geq \frac{2^n(n-2)}{3\log_2 (n+q)} - \frac{n}{3} \;  ,
        \label{formula_lower_complexity}\\
    D(n,q) &\geq \frac{2^n(n-2)}{3(n+q)\log_2 (n+q)} - \frac{n}{3(n+q)} \;  ,
        \label{formula_lower_depth}\\
    W(n,q) &\geq \min(\WC, \WT) \cdot \left(\frac{2^n(n-2)}{3\log_2 (n+q)} - \frac{n}{3} \right) \; .
        \label{formula_lower_quantum_weight}
\end{align}
Also, the following uppper bounds for a reversible circuit
without additional inputs were proved (see~\citealt{my_complexity_and_depth_no_memory}):
\begin{align}
    L(n,0) &\leq \frac{3n2^{n+4}}{\log_2 n - \log_2 \log_2 n - \log_2 \phi(n)}
            \left( 1 + \epsilon_L(n) \right) \;  ,
        \label{formula_upper_complexity_no_mem}\\
    D(n,0) &\leq \frac{n2^{n+5}}{\log_2 n - \log_2 \log_2 n - \log_2 \phi(n)}
            \left( 1 + \epsilon_D(n) \right) \;  ,\\
    W(n,0) &\leq \frac{n2^{n+4} \left( \WC(1 + \epsilon_C(n)) + 2\WT(1 + \epsilon_T(n) \right)}
            {\log_2 n - \log_2 \log_2 n - \log_2 \phi(n)} \;  ,
        \label{formula_upper_quantum_weight_no_mem}
\end{align}
where $\phi(n) < n \mathop / \log_2 n$ is an arbitrarily slowly growing function and
\begin{align*}
    \epsilon_L(n) &= \frac{1}{6\phi(n)} +\left(\frac{8}{3} - o(1)\right)
            \frac{\log_2 n \cdot \log_2 \log_2 n}{n} \;  , \\
    \epsilon_D(n) &= \frac{1}{4\phi(n)} +(4 - o(1))\frac{\log_2 n \cdot \log_2 \log_2 n}{n} \;  , \\
    \epsilon_C(n) &= \frac{1}{2\phi(n)} - \left( \frac{1}{2} - o(1) \right) \cdot \frac{\log_2 \log_2 n }{n} \;  , \\
    \epsilon_T(n) &= (4 - o(1))\frac{\log_2 n \cdot \log_2 \log_2 n}{n} \;  .
\end{align*}

Unfortunately, there are no known upper asymptotic bounds for the functions $L(n,q)$, $D(n,q)$ and $W(n,q)$ in the case when
a reversible circuit can use an unlimited amount of additional inputs for today.
Nevertheless, it has been already showed that in some cases the usage of additional memory in a reversible circuit
consisting of gates from $\Omega_n^2$ allows to reduce its gate complexity and the depth, see%
~\cite{barenco,miller_reducing_complexity, abdessaied_reducing_depth}.

\section{Reducing the gate complexity with the help of additional inputs}\label{section_minimal_gate_complexity}

Lupanov described asymptotically the best synthesis algorithm of a Boolean function in the basis
$\{\,\bar {\phantom x}, \wedge, \vee\,\}$. He proved that any Boolean function of $n$ variables can be
implemented in a circuit with the gate complexity $L \sim 2^n \mathop / n$
and with the total delay no more than $O(n)$, see~\cite{lupanov}.
We will modify Lupanov's method in order to synthesize a reversible circuit, which consists of gates from $\Omega_{n+q}^2$ and
implements a transformation $f \in F(n,q)$ with $q$ additional inputs.

The basis $\{\,\bar {\phantom x}, \oplus, \wedge\,\}$ is functionally complete, therefore it can be used to implement any
transformation $f \in F(n,q)$. Let's express every element of this basis
via a composition of NOT, CNOT and C\textsuperscript{2}NOT gates (see~\ref{pic_basis}).
As we can see, this requires no more than two gates and one additional input for every element of the basis.

\begin{figure}[ht]
    \centering
    \includegraphics{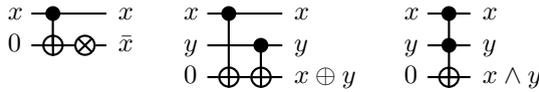}
    \caption
    {
        Implementing elements of the basis $\{\,\bar{\text{ \,}}, \oplus, \wedge\,\}$ with
        compositions of NOT, CNOT and C\textsuperscript{2}NOT gates.
    }\label{pic_basis}
\end{figure}

First, we prove the following lemma about the gate complexity of a reversible circuit implementing all conjunctions
of $n$ variables of the form $x_1^{a_1} \wedge \cdots \wedge x_n^{a_n}$, $a_i \in \BB$.
\begin{lemma}\label{lemma_complexity_of_all_conjunctions_of_n_variables}
All conjunctions of $n$ variables of the form $x_1^{a_1} \wedge \cdots \wedge x_n^{a_n}$, $a_i \in \BB$,
can be implemented in a reversible circuit $\frS_n$, which consists of gates from $\Omega_{n+q}^2$,
with the gate complexity $L(\frS_n) \sim 2^n$ and with $q(\frS_n) \sim 2^n$ additional inputs.
\end{lemma}
\begin{proof}
First step is obtaining inversions of all input variables: $\bar x_i$, $1 \leq i \leq n$.
This can be done using $L_1 = 2n$ NOT and CNOT gates and $q_1 = n$ additional inputs.

We construct our reversible circuit $\frS_n$ this way: using circuits $\frS_{\lceil n/2 \rceil}$
and $\frS_{\lfloor n/2 \rfloor}$, we implement all conjunctions of the first $\lceil n/2 \rceil$
and the last $\lfloor n/2 \rfloor$ variables (see~\ref{pic_conjunctions_complexity}).
After this we implement conjunctions of outputs of the circuit $\frS_{\lceil n/2 \rceil}$ with outputs of the circuit
$\frS_{\lfloor n/2 \rfloor}$. This can be done using $L_2 = 2^n$ C\textsuperscript{2}NOT gates and $q_2 = 2^n$ additional inputs.

Hence, the following equation holds:
$$
    L(\frS_n) \sim q(\frS_n) \sim 2^n + 2L(\frS_{n/2}) \sim 2^n \;  .
$$
\end{proof}

\begin{figure}[ht]
    \centering
    \includegraphics{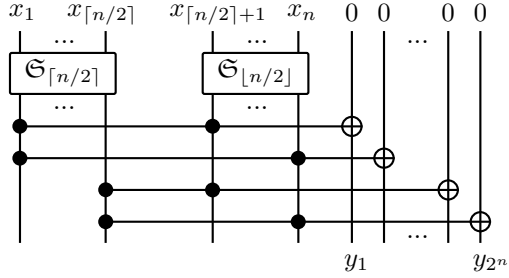}
    \caption
    {
        The structure of a reversible circuit $\frS_n$ implementing conjunctions
        of $n$ variables with the minimal gate complexity.
    }\label{pic_conjunctions_complexity}
\end{figure}

Now we can prove the first theorem of the paper.
\begin{theorem}\label{theorem_complexity_upper_with_memory}
    $$
        L(n,q_0) \lesssim 2^n, \text{ \,if \, } q_0 \sim n 2^{n-\lceil n \mathop / \phi(n)\rceil} \;  ,
    $$
    where $\phi(n) \leq n \mathop / (\log_2 n + \log_2 \psi(n))$
    and $\psi(n)$ are arbitrarily slowly growing functions.
\end{theorem}
\begin{proof}
    We will describe a new synthesis algorithm \textbf{A1}, which is similar to the Lupanov's method
    and whose main goal is the reduction of the gate complexity with the help of additional inputs.

    Let's consider a transformation $f\colon \BB^n \to \BB^n$. It can be represented as follows:
    \begin{multline}\label{formula_function_decomposition_by_last_variables}
        f(\vv x) = \bigoplus_{a_{k+1}, \cdots, a_n \in \BB} {x_{k+1}^{a_{k+1}} \wedge \cdots \wedge x_n^{a_n}} \wedge \\
            \wedge f(\langle x_1, \cdots, x_k, a_{k+1}, \cdots, a_n \rangle) \;  .
    \end{multline}

    Each of $2^{n-k}$ Boolean transformations
    $$
        f(\langle x_1, \cdots, x_k, a_{k+1}, \cdots, a_n \rangle) = f_i(\langle x_1, \cdots, x_k \rangle)  \; ,
    $$
    where $\sum_{j=1}^{n-k} {a_{k+j} 2^{j-1}} = i$,
    is a Boolean transformation $\BB^k \to \BB^n$ and can be represented as the system of $n$ coordinate functions
    $f_{i,j}(\vv x)$, $\vv x \in \BB^k$, $1 \leq j \leq n$.

    The value of every coordinate function $f_{i,j}(\vv x)$ can be calculated with the help of an analogue of
    a disjunctive normal form:
    \begin{equation}\label{formula_analog_sdnf}
        f_{i,j}(\vv x) = \bigoplus_{
            \substack{\boldsymbol \sigma \in \BB^k \\f_{i,j}(\boldsymbol \sigma) = 1}}
            x_1^{\sigma_1} \wedge \cdots \wedge x_k^{\sigma_k} \;  .
    \end{equation}
    
    All $2^k$ conjunctions of the form $x_1^{\sigma_1} \wedge \cdots \wedge x_k^{\sigma_k}$ can be divided into the groups
    with no more than $s$ conjunctions in each. The number of such groups is $p = \lceil 2^k \mathop / s \rceil$.
    Using conjunctions of a single group, we can construct no more than $2^s$ Boolean functions
    by the formula~\eqref{formula_analog_sdnf}.

    Let $G_i$ be the set of all Boolean functions that can be constructed with the help of conjunctions of an $i$-th group,
    $1 \leq i \leq p$. Then $|G_i| \leq 2^s$.
    Therefore, we can rewrite equation~\eqref{formula_analog_sdnf} as follows:
    \begin{equation}\label{formula_analog_sdnf_improved}
        f_{i,j}(\vv x) = \bigoplus_{
            \substack{t=1 \cdots p\\ g_{j_t} \in G_t\\ 1 \leq j_t \leq |G_t|}} g_{j_t}(\vv x) \;  .
    \end{equation}
    
    Note that all Boolean functions of a group $G_i$ can be implemented, using a similar technique as in the%
    ~\ref{lemma_complexity_of_all_conjunctions_of_n_variables}.
    From the~\ref{pic_conjunctions_complexity} we can see that all C\textsuperscript{2}NOT gates
    will be simply replaced by compositions
    of two CNOT gates. Thus, $L \sim 2^{s+1}$ CNOT gates and $q \sim 2^s$ additional inputs are required for this part
    in total.

    The synthesis algorithm \textbf{A1} constructs a reversible circuit $\frS$ implementing the transformation $f$%
    ~\eqref{formula_function_decomposition_by_last_variables}
    from the following sub-circuits (see~\ref{pic_five_subschemes}):

    \begin{figure}[ht]
        \centering
        \includegraphics{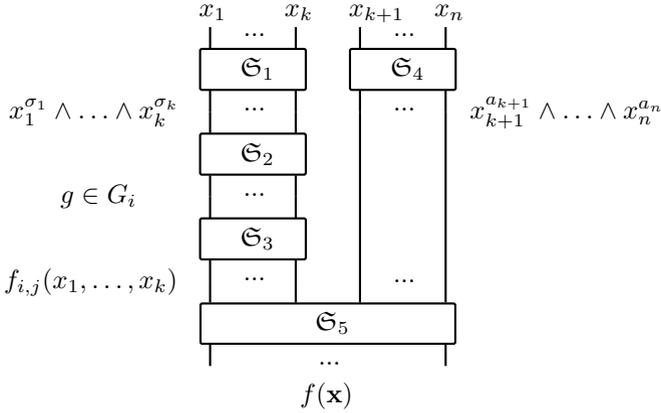}
        \caption
        {
            The structure of a reversible circuit $\frS$ produced by the synthesis algorithm \textbf{A1}.
        }\label{pic_five_subschemes}
    \end{figure}

    \begin{enumerate}
        \item
            Sub-circuit $\frS_1$ implementing all conjunctions of the first $k$ variables $x_i$
            by the~\ref{lemma_complexity_of_all_conjunctions_of_n_variables}
            with the gate complexity $L_1 \sim 2^k$ and with $q_1 \sim 2^k$ additional inputs.
            The sub-circuit almost completely consists of C\textsuperscript{2}NOT gates
            (the number of other gates is negligible).

        \item
            Sub-circuit $\frS_2$ implementing all Boolean functions $g \in G_i$ for all $i \in \ZZ_p$
            by the formula~\eqref{formula_analog_sdnf} with the gate complexity $L_2 \sim p2^{s+1}$ and with
            $q_2 \sim p2^s$ additional inputs (see the note above about the implementation
            of all Boolean functions of a group $G_i$).
            The sub-circuit consists only of CNOT gates.
        
        \item
            Sub-circuit $\frS_3$ implementing all $n2^{n-k}$ coordinate functions $f_{i,j}(\vv x)$,
            $i \in \ZZ_{2^{n-k}}$, $j \in \ZZ_n$, by the formula~\eqref{formula_analog_sdnf_improved}
            with the gate complexity $L_3 \leq pn 2^{n-k}$ and with $q_3 = n 2^{n-k}$ additional inputs.
            The sub-circuit consists only of CNOT gates.

        \item
            Sub-circuit $\frS_4$ implementing all conjunctions of the last $(n-k)$ variables $x_i$
            by the~\ref{lemma_complexity_of_all_conjunctions_of_n_variables}
            with the gate complexity $L_4 \sim 2^{n-k}$ and with $q_4 \sim 2^{n-k}$ additional inputs.
            The sub-circuit almost completely consists of C\textsuperscript{2}NOT gates
            (the number of other gates is negligible).
            
        \item
            Sub-circuit $\frS_5$ implementing the transformation $f$
            by the formula~\eqref{formula_function_decomposition_by_last_variables}
            with the gate complexity $L_5 \leq n 2^{n-k}$ and with $q_5 = n$ additional inputs.
            The sub-circuit consists only of C\textsuperscript{2}NOT gates.
    \end{enumerate}

    We are seeking the values of parameters $k$ and $s$ that satisfy the following conditions:
    $$
        \left\{
            \begin{array}{lr}
                s = n - 2k \;, & \\
                k = \lceil n \mathop / \phi(n) \rceil \;, & \text{where $\phi(n)$ is a growing function,} \\
                1 \leq s < n \;, & \\
                1 \leq k < n \mathop / 2 \;, & \\
                2^k \mathop / s \geq \psi(n) \;, & \text{where $\psi(n)$ is a growing function.}
            \end{array}
        \right.
    $$
    In this case $p = \lceil 2^k \mathop / s \rceil \sim 2^k \mathop / s$
    and $2^{\lceil n \mathop / \phi(n) \rceil} \geq s\psi(n)$.
    This implies that for any growing function $\phi(n) \leq n \mathop / (\log_2n + \log_2 \psi(n))$
    the values of parameters $k$ and $s$ will satisfy the conditions above.

    Summing up gate complexities and the number of additional inputs of sub-circuits $\frS_1$--$\frS_5$,
    we obtain the following bounds:
    \begin{gather*}
        L(\frS) \sim 2^k + p2^{s+1} + pn 2^{n-k} + 2^{n-k} + n 2^{n-k}  \;  , \\
        L(\frS) \sim 2^k + \frac{2^{n-k+1}}{s} + \frac{n2^n}{s}  \;  , \\
        q(\frS) \sim 2^k + p2^s + n2^{n-k} + 2^{n-k} + n \sim 2^k + \frac{2^{n-k}}{s} + n 2^{n-k}  \;  .
    \end{gather*}
    
    Hence, if $k = \lceil n \mathop / \phi(n) \rceil$ and $s = n - 2k$,
    where $\phi(n) \leq n \mathop / (\log_2n + \log_2 \psi(n))$ and $\psi(n)$ are growing functions,
    the following equations hold:
    \begin{gather*}
        L(\frS) \sim 2^{\lceil n \mathop / \phi(n) \rceil} + \frac{2^{n+1}}{n(1-o(1))2^{\lceil n \mathop / \phi(n) \rceil}}
            + \frac{n2^n}{n(1-o(1))} \sim 2^n   \;  , \\
        q(\frS) \sim 2^{\lceil n \mathop / \phi(n) \rceil} + \frac{2^n}{n(1-o(1))2^{\lceil n \mathop / \phi(n) \rceil}}
            + \frac{n 2^n}{2^{\lceil n \mathop / \phi(n) \rceil}} \sim \frac{n 2^n}{2^{\lceil n \mathop / \phi(n) \rceil}} \; .
    \end{gather*}

    Since the synthesis algorithm \textbf{A1} can produce a reversible circuit $\frS$ for any
    Boolean transformation $f \in F(n,q)$, we can state that
    $L(n,q_0) \leq L(\frS) \sim 2^n$, if $q_0 \sim n2^{n - \lceil n \mathop / \phi(n) \rceil}$,
    where $\phi(n) \leq n \mathop / (\log_2 n + \log_2 \psi(n))$ and $\psi(n)$ are arbitrarily slowly growing functions.
\end{proof}

\begin{theorem}
    $$
        L(n,q_0) \asymp 2^n, \text{ \,if \, } q_0 \sim n 2^{n-\lceil n \mathop / \phi(n)\rceil} \;  ,
    $$
    where $\phi(n) \leq n \mathop / (\log_2 n + \log_2 \psi(n))$
    and $\psi(n)$ are arbitrarily slowly growing functions.
\end{theorem}
\begin{proof}
    Follows from the lower bound~\eqref{formula_lower_complexity} for the function $L(n,q)$
    and from the~\ref{theorem_complexity_upper_with_memory}.
\end{proof}
    
Now we can upper bound the quantum weight of a reversible circuit with additional inputs.
\begin{theorem}\label{theorem_quantum_weight_upper_with_memory}
    $$
        W(n,q_0) \lesssim \WC \cdot 2^n + \WT \cdot n2^{n - \lceil n \mathop / \phi(n)\rceil},
            \text{ \,if \, } q_0 \sim n 2^{n-\lceil n \mathop / \phi(n)\rceil}  \;  ,
    $$
    where $\phi(n) \leq n \mathop / (\log_2 n + \log_2 \psi(n))$
    and $\psi(n)$ are arbitrarily slowly growing functions.
\end{theorem}
\begin{proof}
    To prove the bound of the theorem, we should count the number of NOT, CNOT and
    C\textsuperscript{2}NOT gates in a reversible circuit produced by the synthesis algorithm \textbf{A1}.

    From the description of the algorithm we can see that
    \begin{align*}
        \LC_1 & = O(k) \; ,          & \LT_1 &\sim 2^k \; ,\\
        \LC_2 & \sim p2^{s+1} \; ,   & \LT_2 &= 0 \; ,\\
        \LC_3 & \leq pn 2^{n-k} \; , & \LT_3 &= 0 \; ,\\
        \LC_4 & = O(n-k) \; ,        & \LT_4 &\sim 2^{n-k} \; ,\\
        \LC_5 & = 0 \; ,             & \LT_5 &\leq n 2^{n-k} \; .
    \end{align*}
    Providing $k = \lceil n \mathop / \phi(n) \rceil$ and $s = n - 2k$,
    where $\phi(n) \leq n \mathop / (\log_2n + \log_2 \psi(n))$ and $\psi(n)$ are growing functions,
    we obtain the following upper bounds:
    \begin{align*}
        \LC(n,q_0) &\lesssim O(k) + p2^{s+1} + pn 2^{n-k} + O(n-k)  \;  , \\
        \LC(n,q_0) &\lesssim \frac{2^{k+s+1}}{s} + \frac{n2^n}{s}  \;  , \\
        \LC(n,q_0) &\lesssim \frac{2^{n+1}}{(n-o(n))2^{\lceil n \mathop / \phi(n) \rceil}} + \frac{n2^n}{n - o(n)}
            \sim 2^n  \;  ,\\
        \LT(n,q_0) &\lesssim 2^k + 2^{n-k} + n 2^{n-k} \sim 2^k + n2^{n-k}  \;  , \\
        \LT(n,q_0) &\lesssim 2^{\lceil n \mathop / \phi(n) \rceil} + \frac{n2^n}{2^{\lceil n \mathop / \phi(n) \rceil}}
            \sim \frac{n2^n}{2^{\lceil n \mathop / \phi(n) \rceil}}  \;  .
    \end{align*}
    
    From these upper bounds and the equation~\eqref{formula_quantum_weight_and_gate_complexity_dependency}
    the upper bound for the function $W(n,q_0)$ from the theorem follows.
\end{proof}

Note that in the case, when $\WT = O(\WC) = const$, the following equation holds:
$$
    W(n, q_0) \asymp \LC(n, q_0) \sim L(n, q_0)  \; ,
$$
where $q_0 \sim n 2^{n-\lceil n \mathop / \phi(n) \rceil}$; $\phi(n) \leq n \mathop / (\log_2 n + \log_2 \psi(n))$
and $\psi(n)$ are arbitrarily slowly growing functions.
In other words, the number of C\textsuperscript{2}NOT gates in a reversible circuit
produced by the synthesis algorithm \textbf{A1} is negligible compared to the overall gate complexity.
And from the equations~\eqref{formula_lower_complexity}, \eqref{formula_lower_quantum_weight},
\eqref{formula_upper_complexity_no_mem} and \eqref{formula_upper_quantum_weight_no_mem}
it follows that in a reversible circuit without additional inputs the number of C\textsuperscript{2}NOT gates
is equivalent by the order of magnitude to the overall gate complexity.

\section{Reducing the depth with the help of additional inputs}\label{section_minimal_depth}

We described the synthesis algorithm \textbf{A1}, whose main goal was the gate complexity reduction with the help of
additional inputs.
However, we can use a similar technique to reduce the depth of a reversible circuit.
Let's denote such an algorithm as \textbf{A2}.
The essence of the synthesis algorithm \textbf{A2} is the copying of the value from an output to the additional inputs
with the logarithmic depth (see~\ref{pic_parallel_copy}).
After this we can perform a desired operation with the depth equal to one.
All we need to do is to copy the value a sufficient number of times.

\begin{figure}[ht]
    \centering
    \includegraphics{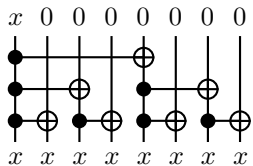}
    \caption
    {
        Copying the value $x$ to the additional inputs with the logarithmic depth.
    }\label{pic_parallel_copy}
\end{figure}

First, we prove the following lemma about the depth of a reversible circuit implementing all conjunctions
of $n$ variables of the form $x_1^{a_1} \wedge \cdots \wedge x_n^{a_n}$, $a_i \in \BB$.
\begin{lemma}\label{lemma_depth_of_all_conjunctions_of_n_variables}
All conjunctions of $n$ variables of the form $x_1^{a_1} \wedge \cdots \wedge x_n^{a_n}$, $a_i \in \BB$,
can be implemented in a reversible circuit $\frS_n$, which consists of gates from $\Omega_{n+q}^2$,
with the depth $D(\frS_n) \sim n$, the gate complexity $L(\frS_n) \sim 3 \cdot 2^n$
and with $q(\frS_n) \sim 3 \cdot 2^n$ additional inputs.
\end{lemma}
\begin{proof}
    First step is obtaining inversions of all input variables: $\bar x_i$, $1 \leq i \leq n$.
    This can be done with the depth $D_1 = 2$,
    using $L_1 = 2n$ NOT and CNOT gates and $q_1 = n$ additional inputs.

    We construct our reversible circuit $\frS_n$ in the same way
    as in~\ref{lemma_complexity_of_all_conjunctions_of_n_variables}, using sub-circuits $\frS_{\lceil n/2 \rceil}$
    and $\frS_{\lfloor n/2 \rfloor}$ (see~\ref{pic_conjunctions_depth}).
    Any output of these sub-circuits will be used no more than in $2 \cdot 2^{n/2}$ conjunctions, so
    all conjunctions can be implemented with the depth $D_2 \leq 2 + n/2$,
    using $2^{n+1}$ CNOT gates, $2^n$ C\textsuperscript{2}NOT gates and $q_2 = 3 \cdot 2^n$ additional inputs.

    Hence, the following equations hold:
    \begin{gather*}
        D(\frS_n) \sim \frac{n}{2} + D(\frS_{n/2}) \sim n \;  , \\
        L(\frS_n) \sim 3 \cdot 2^n + 2L(\frS_{n/2}) \sim 3 \cdot 2^n \;  , \\
        q(\frS_n) \sim 3 \cdot 2^n + 2q(\frS_{n/2}) \sim 3 \cdot 2^n \;  .
    \end{gather*}
\end{proof}

\begin{figure}[ht]
    \centering
    \includegraphics{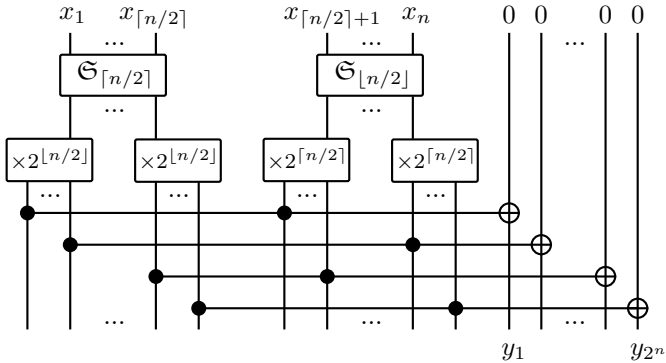}
    \caption
    {
        The structure of a reversible circuit $\frS_n$ implementing all conjunctions
        of $n$ variables with the minimal depth.
    }\label{pic_conjunctions_depth}
\end{figure}

Now we can prove the next theorem of the paper.
\begin{theorem}\label{theorem_depth_upper_with_memory_3n}
    $$
        D(n,q_1) \lesssim 3n, \text{ \,if\quad} q_1 \sim 2^n \; .
    $$
    A reversible circuit $\frS$ with the depth $D(\frS) \sim 3n$ has the gate complexity $L(\frS) \sim 2^{n+1}$
    and the quantum weight $W(\frS) \sim \WC \cdot$ $\cdot 2^{n+1} + \WT \cdot n2^{n - \lceil n \mathop / \phi(n)\rceil}$,
    where $\phi(n) \leq n \mathop / (\log_2 n + \log_2 \psi(n))$ and $\psi(n)$ are arbitrarily slowly growing functions.
\end{theorem}
\begin{proof}
    We will describe the synthesis algorithm \textbf{A2}, which is similar to the synthesis algorithm \textbf{A1}
    and whose main goal is the reduction of the depth with the help of additional inputs.

    Let's consider a transformation $f\colon \BB^n \to \BB^n$.
    It can be represented by the formulae~\eqref{formula_function_decomposition_by_last_variables}--%
    \eqref{formula_analog_sdnf_improved}, see pages~\pageref{formula_function_decomposition_by_last_variables}--%
    \pageref{formula_analog_sdnf_improved}.

    Note that all Boolean functions of a group $G_i$ can be implemented, using a similar technique
    as in the~\ref{lemma_depth_of_all_conjunctions_of_n_variables}. This requires $L \sim 3 \cdot 2^s$ CNOT gates
    ($2^{s+1}$ gates are used to copy values to the additional inputs) and $q \sim 2^{s+1}$ additional inputs
    (see~\ref{pic_all_g_i_depth}).    
    
    \begin{figure}[ht]
        \centering
        \includegraphics{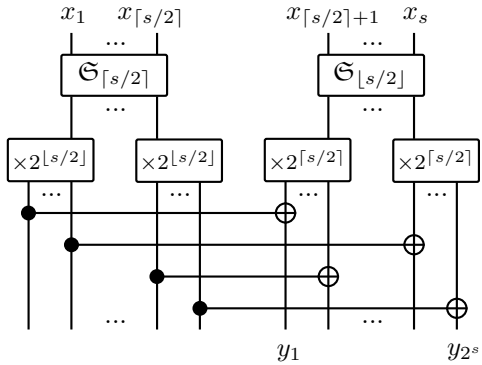}
        \caption
        {
            The structure of a reversible circuit implementing all Boolean functions $g \in G_i$
            with the minimal depth.
        }\label{pic_all_g_i_depth}
    \end{figure}
    
    The synthesis algorithm \textbf{A2} constructs a reversible circuit $\frS$ implementing the transformation $f$%
    ~\eqref{formula_function_decomposition_by_last_variables}
    from the following sub-circuits (see~\ref{pic_six_subschemes}):

    \begin{figure}[ht]
        \centering
        \includegraphics{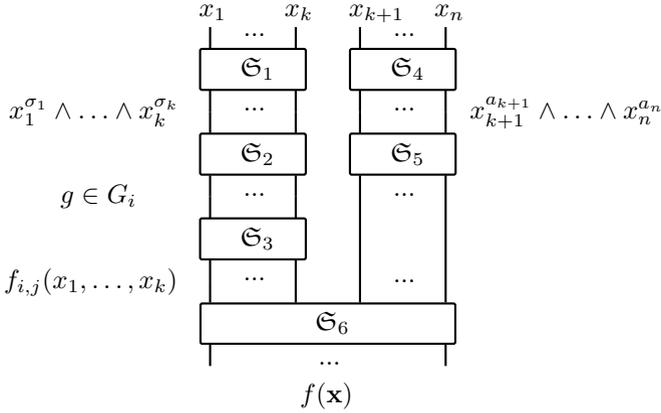}
        \caption
        {
            The structure of a reversible circuit $\frS$ produced by the synthesis algorithm \textbf{A2}.
        }\label{pic_six_subschemes}
    \end{figure}

    \begin{enumerate}
        \item
            Sub-circuit $\frS_1$ implementing all conjunctions of the first $k$ variables $x_i$
            by the~\ref{lemma_depth_of_all_conjunctions_of_n_variables}
            with the depth $D_1 \sim k$, the gate complexity $L_1 \sim 3 \cdot 2^k$ and with $q_1 \sim 3 \cdot 2^k$
            additional inputs. The sub-circuit contains $2^k$ C\textsuperscript{2}NOT gates.

        \item
            Sub-circuit $\frS_2$ implementing all Boolean functions $g \in G_i$ for all $i \in \ZZ_p$
            by the formula~\eqref{formula_analog_sdnf} with the depth $D_2 \sim s$,
            the gate complexity $L_2 \sim 3p2^s$ and with
            $q_2 \sim p2^{s+1}$ additional inputs (see the note above about the implementation
            of all Boolean functions of a group $G_i$).
            The sub-circuit consists only of CNOT gates.

        \item
            Sub-circuit $\frS_3$ implementing all $n2^{n-k}$ coordinate functions $f_{i,j}(\vv x)$,
            $i \in \ZZ_{2^{n-k}}$, $j \in \ZZ_n$, by the formula~\eqref{formula_analog_sdnf_improved}.
            The feature of this sub-circuit is that a Boolean function $g \in G_t$ can be used more than once.
            The maximum usage count for a function $g$ is $n2^{n-k}$.            
            So, first of all we need to copy the values from outputs of the sub-circuit $\frS_2$
            for all such Boolean functions.
            This can be done with the depth equal to $\lceil n-k +\log_2 n\rceil$, using no more than $pn2^{n-k}$ gates
            and $pn2^{n-k}$ additional inputs (see~\ref{pic_parallel_copy}).
            After this, we implement XOR function of obtained outputs with the depth equal to $\lceil\log_2 p\rceil$,
            the gate complexity equal to $(p-1)n2^{n-k}$
            and without additional inputs (see~\ref{pic_logarithmic_xor}).
            Therefore, the sub-circuit $\frS_3$ has the depth $D_3 \sim n-k + \log_2 p$,
            the gate complexity $L_3 \sim (2p-1)n 2^{n-k}$ and $q_3 \sim (p-1)n 2^{n-k}$ additional inputs.
            It consists only of CNOT gates.

        \item
            Sub-circuit $\frS_4$ implementing all conjunctions of the last $(n-k)$ variables $x_i$
            by the~\ref{lemma_depth_of_all_conjunctions_of_n_variables}
            with the depth $D_4 \sim (n-k)$, the gate complexity $L_4 \sim 3 \cdot 2^{n-k}$ and with $q_4 \sim 3 \cdot 2^{n-k}$
            additional inputs. The sub-circuit contains $2^{n-k}$ C\textsuperscript{2}NOT gates.

        \item
            Sub-circuit $\frS_5$, which is needed to copy $(n-1)$ times every output of the sub-circuit $\frS_4$.
            This can be done with the depth $D_5 \sim \log_2 n$, the gate complexity $L_5 = (n-1)2^{n-k}$
            and with $q_5 = (n-1)2^{n-k}$ additional inputs.
            The sub-circuit consists only of CNOT gates.

        \item        
            Sub-circuit $\frS_6$ implementing the transformation $f$
            by the formula~\eqref{formula_function_decomposition_by_last_variables}.
            The structure of the sub-circuit is as follows: all $n2^{n-k}$ coordinate functions $f_{i,j}(\vv x)$
            are grouped by $2^{n-k}$ functions ($n$ groups in total, which correspond to $n$ outputs of the transformation $f$).
            Functions in a group are again grouped by two.
            For every pair of functions we implement a conjunction of the corresponding outputs of sub-circuits
            $\frS_3$ and $\frS_5$, using 2 C\textsuperscript{2}NOT gates and one additional input for
            storing an intermediate result (see~\ref{pic_s_5_logarithmic_depth}).
            Thus, this part of the sub-circuit has the depth equal to 2, requires $n2^{n-k}$ C\textsuperscript{2}NOT gates and
            $n2^{n-k-1}$ additional inputs.
            After this, in every of $n$ groups of obtained outputs we implement XOR function with the logarithmic depth
            (see~\ref{pic_logarithmic_xor} and~\ref{pic_s_5_logarithmic_depth}).
            This part of the sub-circuit has the depth equal to $n-k-1$, requires $n(2^{n-k-1} - 1)$ CNOT gates and
            doesn't require additional inputs, because we can use the existing outputs.
            
            Summing up, the sub-circuit $\frS_6$ has the depth $D_5 \sim n-k$, the gate complexity $L_6 \sim 3n 2^{n-k-1}$
            and $q_6 \sim n 2^{n-k-1}$ additional inputs.
    \end{enumerate}
    \begin{figure}[ht]
        \centering
        \includegraphics{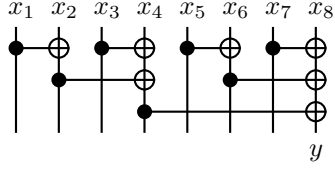}
        \caption
        {
            Implementing the function $y = x_1 \oplus \cdots \oplus x_8$ in a reversible circuit with the logarithmic depth
            (this is a part of the sub-circuit $\frS_3$ produced by the synthesis algorithm \textbf{A2}).
        }\label{pic_logarithmic_xor}
    \end{figure}
    \begin{figure}[ht]
        \centering
        \includegraphics{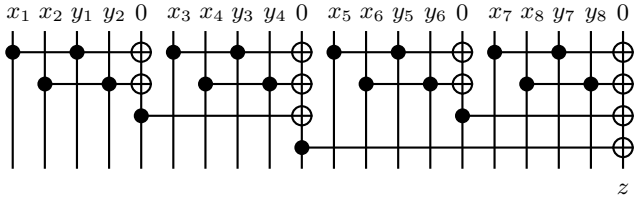}
        \caption
        {
            Implementing the function $z = \bigoplus_{i=1}^8{x_i \wedge y_i}$ in a reversible circuit with the logarithmic depth
            (this is a part of the sub-circuit $\frS_6$ produced by the synthesis algorithm \textbf{A2}).
        }\label{pic_s_5_logarithmic_depth}
    \end{figure}

    Note that the sub-circuits $\frS_1$--$\frS_3$ and $\frS_4$--$\frS_5$ can work in parallel,
    because they use disjoint subsets of the inputs $x_1, \cdots, x_n$.

    We are seeking the values of parameters $k$ and $s$ that satisfy the following conditions:
    $$
        \left\{
            \begin{array}{lr}
                k + s = n \;, & \\
                1 \leq k < n \;, & \\
                1 \leq s < n \;, & \\
                2^k \mathop / s \geq \psi(n) \;, & \text{where $\psi(n)$ is a growing function.}
            \end{array}
        \right.
    $$
    In this case $p = \lceil 2^k \mathop / s \rceil \sim 2^k \mathop / s$.

    Summing up depths, gate complexities and the number of additional inputs for all sub-circuits $\frS_1$--$\frS_6$,
    we obtain the following bounds for the circuit $\frS$ parameters.
    
    The depth:
    \begin{gather}
        D(\frS) \sim \max(k + s + n - k + \log_2 p \;;\; n-k + \log_n) + n - k  \; ,  \notag \\
        D(\frS) \sim 2n + s \; . \label{formula_depth_general_upper_bound}
    \end{gather}
    
    The gate complexity:
    \begin{multline*}
        L(\frS) \sim 3 \cdot 2^k + 3p2^s + (2p-1)n2^{n-k} + \\
            + 3 \cdot 2^{n-k} + n2^{n-k} + 3n2^{n-k-1}  \; ,
    \end{multline*}
    \begin{equation}
        L(\frS) \sim 3 \cdot \frac{2^n}{2^s} + \frac{3 \cdot 2^n}{s} + \frac{n2^{n+1}}{s} \sim \frac{n2^{n+1}}{s}  \; .
            \label{formula_complexity_for_depth_general_upper_bound}
    \end{equation}
    
    The number of additional inputs:
    \begin{gather}
        q(\frS) \sim 3 \cdot 2^k + p2^{s+1} + pn2^{n-k} + 3 \cdot 2^{n-k} + n2^{n-k} + n2^{n-k-1}  \; , \notag \\
        q(\frS) \sim 3 \cdot \frac{2^n}{2^s} + \frac{2^{n+1}}{s} + \frac{n2^n}{s} \sim \frac{n2^n}{s}  \; .
            \label{formula_memory_for_depth_general_upper_bound}
    \end{gather}
       
    From the description of the synthesis algorithm \textbf{A2} we can see that
    \begin{align*}
        \LC_1 & \sim 2^{k+1} \; ,         & \LT_1 &\sim 2^k \; ,\\
        \LC_2 & \sim 3p2^s \; ,           & \LT_2 &= 0 \; ,\\
        \LC_3 & \sim pn 2^{n-k+1} \; ,    & \LT_3 &= 0 \; ,\\
        \LC_4 & \sim 2^{n-k+1} \; ,       & \LT_4 &\sim 2^{n-k} \; ,\\
        \LC_5 & \sim n2^{n-k} \; ,        & \LT_5 &= 0 \; ,\\
        \LC_6 & \sim n2^{n-k-1} \; ,      & \LT_6 &= n 2^{n-k} \; .
    \end{align*}
    This implies that
    \begin{align}
        \LC(\frS) &\sim 2^{k+1} + \frac{n2^{n+1}}{s} \sim \frac{n2^{n+1}}{s}  \; ,
            \label{formula_lc_complexity_for_depth_general_upper_bound} \\
        \LT(\frS) &\sim 2^k + n2^{n-k}  \; .
            \label{formula_lt_complexity_for_depth_general_upper_bound}
    \end{align}
        
    Let $k = \lceil n \mathop / \phi(n) \rceil$, where $\phi(n) < n$ is a growing function.
    In this case $s = n - \lceil n \mathop / \phi(n)\rceil$ and
    \begin{multline*}
        2^k \geq s\psi(n) \Rightarrow k \geq \log_2 s + \log_2 \psi(n) \Rightarrow \\
            \Rightarrow \phi(n) \leq \frac{n}{\log_2 s + \log_2 \psi(n) - 1}  \; .
    \end{multline*}
    We can choose any arbitrarily slowly growing functions $\phi(n) \leq n \mathop / (\log_2 n + \log_2 \psi(n))$
    and $\psi(n)$.

    Hence, we obtain the following bounds for the circuit $\frS$ parameters:
    \begin{gather*}
        D(\frS) \sim 2n + n - \lceil n \mathop / \phi(n)\rceil \sim 3n  \; , \\
        L(\frS) \sim \LC(\frS) \sim \frac{n2^{n+1}}{n - \lceil n \mathop / \phi(n)\rceil} \sim 2^{n+1} \; , \\
        \LT(\frS) \sim 2^{\lceil n \mathop / \phi(n)\rceil} + n2^{n - \lceil n \mathop / \phi(n)\rceil}
            \sim n2^{n - \lceil n \mathop / \phi(n)\rceil} \; , \\
        q(\frS) \sim \frac{n2^n}{n - \lceil n \mathop / \phi(n)\rceil} \sim 2^n \; . \\
    \end{gather*}

    Since the synthesis algorithm \textbf{A2} can produce a reversible circuit $\frS$ for any
    Boolean transformation $f \in F(n,q)$, we can state that $D(n,q_1) \leq D(\frS) \sim 3n$, if $q_1 \sim 2^n$.
    
    Also we can state that $L(\frS) \sim 2^{n+1}$ and
    $W(\frS) \sim \WC \cdot 2^{n+1} + \WT \cdot n2^{n - \lceil n \mathop / \phi(n)\rceil}$,
    where $\phi(n) \leq n \mathop / (\log_2 n + \log_2 \psi(n))$ and $\psi(n)$ are arbitrarily slowly growing functions.
\end{proof}

Finally, we prove the last theorem of the paper.
\begin{theorem}\label{theorem_depth_upper_with_memory_2n}
    $$
        D(n,q_2) \lesssim 2n, \text{ \,if\quad} q_2 \sim \phi(n)2^n   \; ,
    $$
    where $\phi(n) < n$ is an arbitrarily slowly growing function.    
    A reversible circuit $\frS$ with the depth $D(\frS) \sim 2n$ has the gate complexity $L(\frS) \sim \phi(n)2^{n+1}$
    and the quantum weight $W(\frS) \sim \WC \cdot$ $\cdot \phi(n)2^{n+1} + \WT \cdot 2^{n - \lceil n \mathop / \phi(n)\rceil}$.
\end{theorem}
\begin{proof}
    Proof is based on the proof of the previous theorem.

    Let $s = \lceil n \mathop / \phi(n) \rceil$, where $\phi(n) < n$ is a growing function.
    In this case $k = n - \lceil n \mathop / \phi(n)\rceil$ and
    $$
        \psi(n) \leq \frac{2^k}{s} \leq \frac{\phi(n)2^{n-o(n)}}{n}  \; . 
    $$
    We can see that we always able to choose a growing function $\psi(n)$ for any growing function $\phi(n) < n$.

    From the equations \eqref{formula_depth_general_upper_bound}--\eqref{formula_lt_complexity_for_depth_general_upper_bound}
    it follows that for these values of $k$ and $s$ the following equations hold:
    \begin{gather*}
        D(\frS) \sim 2n + \lceil n \mathop / \phi(n) \rceil \sim 2n  \; , \\
        L(\frS) \sim \LC(\frS) \sim \frac{n2^{n+1}}{\lceil n \mathop / \phi(n) \rceil} \sim \phi(n)2^{n+1} \; , \\
        \LT(\frS) \sim 2^{n - \lceil n \mathop / \phi(n)\rceil} + n2^{\lceil n \mathop / \phi(n) \rceil}
            \sim 2^{n - \lceil n \mathop / \phi(n)\rceil} \; , \\
        q(\frS) \sim \frac{n2^n}{\lceil n \mathop / \phi(n) \rceil} \sim \phi(n)2^n \; . \\
    \end{gather*}

    Since the synthesis algorithm \textbf{A2} can produce a reversible circuit $\frS$ for any
    Boolean transformation $f \in F(n,q)$, we can state that $D(n,q_2) \leq D(\frS) \sim 2n$, if $q_2 \sim \phi(n)2^n$,
    where $\phi(n) < n$ is an arbitrarily slowly growing function.
    
    Also, $L(\frS) \sim \phi(n)2^{n+1}$ and $W(\frS) \sim \WC \cdot \phi(n)2^{n+1}
        + \WT \cdot$ $\cdot 2^{n - \lceil n \mathop / \phi(n)\rceil}$.
\end{proof}

\section{Conclusion}
We have discussed the problem of the general synthesis of a reversible circuit,
which consists of NOT, CNOT and C\textsuperscript{2}NOT gates and which has additional inputs,
with the lowest possible gate complexity and depth.
We have studied the Shannon gate complexity function $L(n,q)$, the depth function $D(n,q)$ and
the quantum weight function $W(n,q)$ for a reversible circuit with additional inputs,
which implements a transformation $f\colon \BB^n \to \BB^n$ from the set $F(n,q)$.

The main result of the paper is the following claim.
\begin{claim}
    The usage of additional memory in a reversible circuit consisting of NOT, CNOT and C\textsuperscript{2}NOT gates
    almost always allows to reduce its gate complexity, the depth and the quantum weight.
\end{claim}
The proof of the claim follows from the
Theorems~\bare\ref{theorem_complexity_upper_with_memory}--\bare\ref{theorem_depth_upper_with_memory_2n}
and the lower bounds~\eqref{formula_lower_complexity}--\eqref{formula_lower_quantum_weight}.

Solving the problem of the reversible logic synthesis, one should find a compromise between the gate complexity,
the depth (working time) and the amount of used memory (additional inputs) of a reversible circuit.
Unfortunately, we were not able to establish good upper and lower bounds for the depth function $D(n,q)$,
which would be asymptotically equivalent to each other by the order of magnitude.
However, the obtained bounds are sufficient to prove the main claim of the paper.

Further research should establish more precise relationship of a reversible circuit's parameters
from the number of additional inputs in the circuit. We hope that the paper will be the first step in this direction.

\begin{acknowledge}
  The reported study was partially supported by RFBR, research project No. 16-01-00196~A.
\end{acknowledge}

\bibliography{cc-journals,\jobname}

\end{document}